\documentclass[10pt, twocolumn]{IEEEtran}
\usepackage{epsfig, psfrag, amsmath, amssymb, amsfonts, latexsym, eucal, balance, hyperref, indentfirst, graphicx, bookmark}
\usepackage[tight,footnotesize]{subfigure}
\newtheorem{thm}{Theorem}%[section]

\newtheorem{lem}{Lemma}

\newtheorem{defn}{Definition}
\newtheorem{rem}{Remark}

\begin{document}

%%%%%%%%%%%%%%%%%%%%%%%%%%%%%%%%%%%%%%%%%%%%%%%%%%%%%%%%%%%%%%%%%%%%%%%%%%%%%%%%%%%%%%%%%%%%%%%%%%%%%%%%%%%%
\title{Competition Between Wireless Service Providers: Pricing, Equilibrium and Efficiency}
\author{Feng Zhang and Wenyi Zhang\\
Department of Electronic Engineering and Information Science\\ University of Science and Technology of China, Hefei, 230027, China\\ Email: wenyizha@ustc.edu.cn}

\maketitle
\thispagestyle{empty}
\pagestyle{empty}

\begin{abstract}
As the communication network is in transition towards a commercial one controlled by service providers (SP) , the present paper considers a pricing game in a communication market covered by several wireless access points sharing the same spectrum and analyzes two business models: monopoly (APs controlled by one SP) and oligopoly (APs controlled by different SPs). We use a Stackelberg
game to model the problem: SPs are the leader(s) and end users are the followers. We prove, under certain conditions, the existence and uniqueness of Nash equilibrium for both models and derive their expressions. In order to compare the impact of different business models on social welfare and SPs' profits, we define two metrics: PoCS (price of competition on social welfare) and PoCP (price of competition on profits). For symmetric cross-AP interferences, the tight lower bound of PoCS is $3/4$, and that of PoCP is $1$.
\end{abstract}

%\begin{keywords}
%Duopoly, Economic model, Game theory, Market demand function,
%Monopoly, Nash equilibrium, Pricing, Stackelberg game, Wardrop
%equilibrium, Wireless access networks
%\end{keywords}

%\newpage
%\setcounter{page}{1}
\makeatletter
    \newcommand{\rmnum}[1]{\romannumeral #1}
    \newcommand{\Rmnum}[1]{\expandafter\@slowromancap\romannumeral #1@}
\makeatother

\section{\textsc{Introduction}}
\label{sect:introduction}

Resource allocation is a basic issue in communication networks. Decentralized resource allocation has gained extensive attention \cite{Ozdaglar07Incentives}, in which the majority of relevant works treated end users and service providers (SPs) as self-centric agents making decisions to maximize their own utilities. Prices, which are set by SPs, act as the bridge between end users and SPs: SPs set proper prices for the services they provide to achieve network improvement, fairness of resource allocation \cite{Kelly98Rate} or maximization of profits \cite{Basar02Revenue}; on the other hand, end users choose the proper services that maximize their own utilities based on charged prices and the corresponding QoS.

\vspace{10pt}

For wired access networks, pricing models have received extensive research in recent years \cite{Acemoglu2007Competition} \cite{Kelly98Rate}. The congestion condition of a transmission link in wired networks only depends on the flows it carries. As shown in previous works, the pricing problem can be converted to convex optimization problems and can be analyzed using convex analysis techniques. For wireless access networks, however, the study is relatively limited. A main obstacle in wireless access networks is the broadcast nature of wireless medium which induces interferences, thus bringing new challenges to the pricing analysis: congestion is not only caused by the users from the same access network, but also from other networks nearby. As a consequence, pricing models in wireless access networks have mainly concentrated on new scenarios arising from new elements introduced besides the existing networks, such as femtocell \cite{Shetty2009Economics}, WiFi and WiMAX  \cite{Niyato2007Wireless}, thus to some extent avoiding the treatment of interference. In \cite{Vajic09Competition}, the authors analyzed the competition between wireless service providers with context similar to us. However, they assumed that the resources among the SPs are orthogonal and SPs are free from cross-AP interferences. As a consequence, the result that competition among SPs leads to a globally optimal outcome may no longer be valid when cross-AP interferences are considered. The work in \cite{Maille2012Competition} took an initial attempt to analyzing a duopoly pricing model when the interferences among the SPs are not neglected. However, therein the amount of interferences among the SPs are always assumed equal, an assumption overly restrictive for realistic applications. As another limitation in general, most of existing studies only considers the monopoly case or the duopoly case, individually, seldom comparing the two business models together. Consequently, the purpose of our study in this paper is twofold. First, we aim to develop a general framework for analyzing pricing in wireless access networks, with interferences taken into consideration. Second, we aim to compare the impacts of monopoly and duopoly on social welfare and SPs' profits.

\vspace{10pt}

In this paper, we propose a general pricing model for multiple wireless APs, which may be viewed as a generalization of a wired access pricing model \cite{Hayrapetyan2007Network}. We consider a hotspot with a large number of end users covered by several APs sharing the same radio band. We adopt a market demand function to capture user heterogeneity in their willingness-to-pay and examine two business models: monopoly and duopoly, respectively. We formulate such a pricing model as a Stackelberg game: SPs are the leader(s) and end users are the followers, and further use Wardrop's principle to describe users' distribution on APs. We establish the existence and uniqueness of the equilibria of the game for both cases and further analytically characterize the equilibria and the corresponding user distributions, under linear demand function and cross-AP interferences model. In order to compare the impacts of monopoly and duopoly on social welfare and SPs' profits, we define two metrics: PoCS (price of competition on social welfare) which is the ratio of social welfare of monopoly to that of duopoly, and PoCP (price of competition on profits) which is the ratio of profits of monopoly to the sum profits of duopoly. Then for symmetric cross-AP interferences, we establish that PoCS is lower bounded by $3/4$ and PoCP lower bounded by $1$, while both of them are unbounded from above.

\vspace{10pt}

The remaining part of this paper is organized as follows. Section
\ref{sect:system model} describes the system model, formulating wireless APs' pricing problem as a Stackelberg game and setting basic assumptions for subsequent analysis. Section \ref{sect:wardrop equilibrium for users} establishes the existence and uniqueness of user distribution equilibrium following Wardrop's principle for a given price profile. Section \ref{sect:oligopoly analysis} analyzes the price equilibria of both monopoly and
duopoly models, as well as the corresponding user distributions. Section \ref{sect:comparison of monopoly and duopoly}
introduces the concepts of PoCS and PoCP, and analyzes their bounds. Section \ref{sect:numerical} provides some numerical illustrations of our analysis. Finally Section \ref{sect:conclusion} concludes this paper.

\section{\textsc{System Model}}
\label{sect:system model}

\subsection{General Model}
\label{subsect:general model}

Suppose that there is a hotspot covered by $N$ wireless APs which provide access services for a large number of end users. Each AP has an access price $p_i$, $i\in\mathcal{N}$, and the users must take payments before use. Let $x_i$ denote the flow carried on $\mathrm{AP}_i$ and $\mathbf{x}=[x_1,...,x_N]$ denote the vector of flows of all APs. As the congestion level of each AP is related to the flow it serves and that of other APs (due to the broadcast nature of wireless channel), we denote each AP's congestion function as $l_i(\mathbf{x})$. Note that the game we considered in this paper is in large time scale and the congestion function only reflects an averaged congestion level. Then we define the disutility of accessing $\mathrm{AP}_i$ as the sum of
$p_i$ and $l_i(\mathbf{x})$, i.e., $d_i\triangleq p_i+l_i(\mathbf{x})$ (implicitly assuming that all users trade off time and money identically \cite{Hayrapetyan2007Network}). The users are selfish and would access the APs with the least disutility, so the disutility of the market can be expressed as $d\triangleq \min_i d_i$. Because the value of the wireless access service is different for different users, we use a market demand function $D(\cdot)$ to capture this heterogeneity in their willingness-to-pay. For convenience, we consider the inverse of the demand function $u(x)\triangleq D^{-1}(x)$, where $x=\sum_i x_i$. We naturally assume that $D(\cdot)$ is decreasing due to the law of demand \cite{Mankiw2011Principles} which indicates that the total flow of the users is decreasing with the disutility of the market increasing. We further assume that the number of end users is large and each single user's impact on the whole system is negligible, i.e., any user's switching from one AP to another AP does not change the congestion situations of all APs. Thus we can use Wardrop's principle \cite{Wardrop1952Some} to describe users' distribution on all APs. Formally, we define Wardrop equilibrium (WE) as follows.

\begin{defn}[Wardrop Equilibrium]
\label{defn:WE}
Given a set of prices of all APs, $\mathbf{p}$, a flow
distribution $\mathbf{x}^{\mathrm{WE}}$ forms a WE if it satisfies:
\begin{eqnarray}
&1)& \;\; \forall x_i^{\mathrm{WE}}>0, \; d_i=p_i+l_i(\mathbf{x}^{\mathrm{WE}})=u(\sum_j
x_j^{\mathrm{WE}});\label{equ:WE1}\\
&2)& \;\; \forall x_i^{\mathrm{WE}}=0, \; d_i=p_i+l_i(\mathbf{x}^{\mathrm{WE}})\geq u(\sum_j
x_j^{\mathrm{WE}}).\label{equ:WE2}
\end{eqnarray}
\end{defn}

We define the profits of $\mathrm{AP}_i$ as $p_i x_i$ and consumers' surplus as $\int_{0}^{u^{\!-1}\!(d)}\!(u(x)\!-\!d)dx$. Then the social welfare is defined as the sum of APs' profits and consumers' surplus, and can further be written as (at a WE):
\begin{eqnarray}
    \mathrm{SW}(\mathbf{x}^{\mathrm{WE}}) &=& \int_0^{\sum_j x_j^{\mathrm{WE}}}\!\!\!\!\!\!\!\!\!\!\!\! u(x)dx-\sum_j x_j^{\mathrm{WE}} l_j(\mathbf{x}^{\mathrm{WE}}).
\end{eqnarray}

When there exist only two APs, see Figure (\ref{fig:Demand_curve}) for an illustration.

\begin{figure}[!hbp]
\centering
\includegraphics[width=3in]{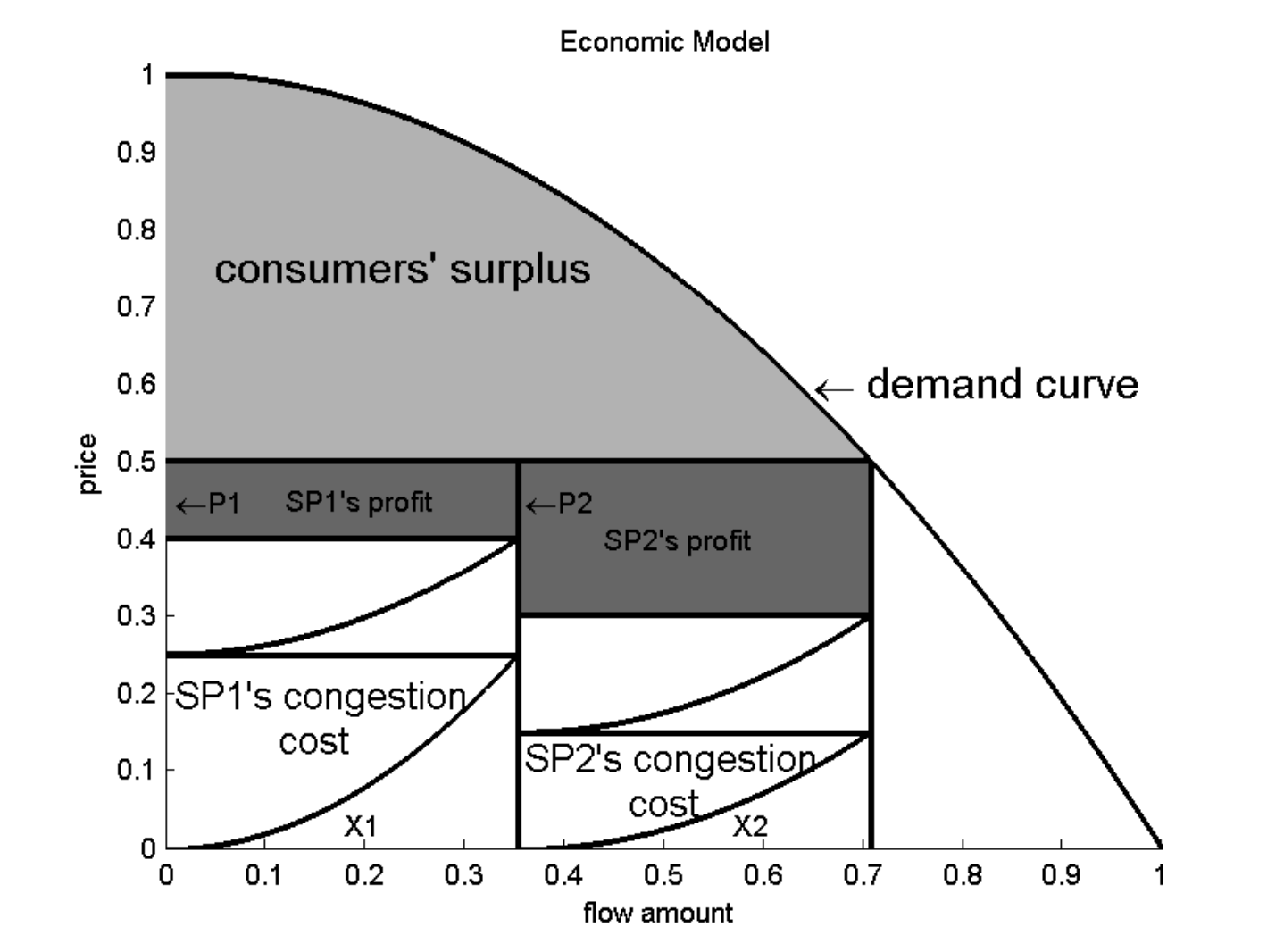}
\centering \caption{Illustration of system model for two-AP case} \label{fig:Demand_curve}
\end{figure}

We formulate this wireless APs pricing model as a Stackelberg game: SPs are the leader(s) and they set prices for the APs first; end users are the followers, they decide whether or not to accept the services and if they do, further decide which APs to access. When all the APs are controlled by a single SP, we have a monopoly market; when the APs are owned by different SPs, we have an oligopoly market. We define monopoly equilibrium (ME) and oligopoly equilibrium (OE) as follows.

\begin{defn}[Monopoly Equilibrium]
A set of prices of all APs, $\mathbf{p}^{\mathrm{ME}}$, forms an ME,
when it satisfies:
\begin{equation}
\mathbf{p}^{\mathrm{ME}} \in \arg \max_{\mathbf{p}} \sum_{i\in \mathcal{I}}
p_ix_i^{\mathrm{WE}}(\mathbf{p}).
\end{equation}
\end{defn}

\begin{defn}[Oligopoly Equilibrium]
A set of prices of all APs, $\mathbf{p}^{\mathrm{OE}}$, forms an OE,
when for any $\mathrm{SP}_i$, the price satisfies:
\begin{equation}
p_i^{\mathrm{OE}} \in \arg \max_{p_i,\mathbf{p}_{-i}^{\mathrm{OE}}}
p_ix_i^{\mathrm{WE}}(p_i,\mathbf{p}_{-i}^{\mathrm{OE}}).
\end{equation}
\end{defn}

\subsection{An SINR Example}
\label{subsect:a motivating model}

A concrete example which motivates and also illustrates the aforementioned general system model is a simple $N$-AP uplink wireless access network. Assume that all users' average transmit power level is a common constant $P$, and that the frequency-flat channel gains of all the user links connecting to the same AP are also identical (averaged in large time scale). Suppose that the average number of users served by $\mathrm{AP}_i$ is $X_i$ and the channel gain from the users of $\mathrm{AP}_j$ to that of $\mathrm{AP}_i$ is $g_{j,i}$ . We focus on the scenario where the number of users is large, and thus the noise as well as each single user's impact on SINR are negligible. Thus the average signal to interference plus noise ratio (SINR) at $\mathrm{AP}_i$ is expressed as (with bandwidth normalized)

\begin{eqnarray}
\mathrm{SINR}_i &\doteq& \frac{g_{i,i}P}{N+g_{i,i}PX_i
+\sum_{j\neq i} g_{j,i}PX_j}\nonumber \\
&\doteq& \frac{g_{i,i}}{g_{i,i}X_i+\sum_{j\neq i} g_{j,i}X_j}.
\end{eqnarray}

\noindent When considering a coding scheme using random Gaussian codebooks and single-user decoding treating interference as noise, a typical user connecting to $\mathrm{AP}_i$ achieves rate (normalized) $r_i=\log(1+\mathrm{SINR}_i)$. Since when the number of users is large and the SINR is low,we have $r_i \approx \mathrm{SINR}_i/\ln 2$. Let each user transmit (on average) a fixed amount $\mu$ of information bits for accessing an AP. As a result, the averaged delay a typical user of $\mathrm{AP}_i$ experiences is

\begin{equation}
l_i(X_1,\cdots,X_N)=\frac{\mu}{r_i}=\mu\frac{g_{i,i}X_i+\sum_{j\neq i}g_{j,i}X_j}{g_{i,i}}\ln2.
\end{equation}

\noindent We then define $x_i\triangleq \mu X_i\ln2$ as the aggregated flow carried on $\mathrm{AP}_i$, and make the substitutions $\forall j\neq i$, $\tilde{g}_{j,i}\triangleq g_{j,i}/g_{i,i}$ and $\tilde{g}_{i,i}=1$ ($G=[\tilde{g}_{i,j}]_{i,j}$).
Thus the disutility of accessing $\mathrm{AP}_i$ is

\begin{equation}
\label{equ:delay function}
d_i = p_i + \sum_{j}\tilde{g}_{j,i}x_j.
\end{equation}

\begin{rem}
In the SINR example above, the linear structure of the disutility function mainly results from two assumptions:
\noindent (1) the problem is considered in a large time scale, and hence the fading effect of wireless channel can be averaged and users' consumption behavior can accord to some market demand curve;
\noindent (2) end users are gathered, and hence their channels to an AP are homogeneous in the long run.
However, as we shall see in the subsequent analysis, the problem is still complicated because both the pricing strategy of SPs and accessing AP strategy of users are intertwined due to the presence of interference, and the solution to the problem exhibits interesting behavior.
\end{rem}

\section{\textsc{Wardrop Equilibrium for Users}}
\label{sect:wardrop equilibrium for users}
In this section, we begin with a study on the existence and uniqueness of user distribution equilibrium, i.e., WE, and we analyze the characteristics of WE.

\subsection{Existence of Wardrop Equilibrium}

Borrowing ideas from proofs in \cite{aashtiani1981equilibria}, we establish the following theorem.

\begin{thm}[Existence of WE]
\label{thm:existence of WE}
For any given positive prices $(p_1,\cdots,p_N)$, if the inverse of the user demand function $u(x)$ is continuous, bounded, non-negative and decreasing, and $\forall i$, $l_i(\mathbf{x})$ is continuous and satisfies $l_i(\mathbf{x})\geq x_i$, then there exits $\mathbf{x}^{\mathrm{WE}}$ constituting a WE.
\end{thm}

\begin{proof}
From Definition \ref{defn:WE}, conditions (\ref{equ:WE1}) and (\ref{equ:WE2}) of WE can be written in the form of a nonlinear complementarity problem (NCP), i.e., for $\forall i \in \{1,\cdots,N\}$, $\mathbf{x}^{\mathrm{WE}}$ satisfies:
\begin{equation}
    \begin{cases}
    x_i^{\mathrm{WE}}[p_i+l_i(\mathbf{x}^{\mathrm{WE}})-u(\sum_{j} x_j^{\mathrm{WE}})] =0\\
    x_i^{\mathrm{WE}}\geq 0, \; p_i+l_i(\mathbf{x}^{\mathrm{WE}})-u(\sum_{j} x_j^{\mathrm{WE}})]\geq 0.
    \end{cases}
\end{equation}
We construct a mapping $\phi = (\phi_1,\cdots,\phi_N)$, where $\phi_i(\mathbf{x}^{\mathrm{WE}})=[x_i^{\mathrm{WE}}-F_i(\mathbf{x}^{\mathrm{WE}})]^+$, $F_i(\mathbf{x}^{\mathrm{WE}}) \triangleq p_i+l_i(\mathbf{x}^{\mathrm{WE}})-u(\sum_{j} x_j^{\mathrm{WE}})$ and $[x]^+ = \max\{x,0\}$. From \cite[Thm. 5.3]{aashtiani1981equilibria}, the existence of the solution of NCP is equivalent to the existence of a fixed point of the mapping $\phi = (\phi_1,\cdots,\phi_N)$. Under the conditions $\forall i, \; l_i(\mathbf{x}^{\mathrm{WE}}) \geq x_i^{\mathrm{WE}}$ and $u(\cdot)$ a decreasing function, for $\forall i \in \{1,\cdots,N\}$, we have $x_i^{\mathrm{WE}}-F_i(\mathbf{x}^{\mathrm{WE}}) \leq u(0)$. As a result, for all $\forall \mathbf{x}^{\mathrm{WE}} \in [0,u(0)]^N$, we have $\phi(\mathbf{x}^{\mathrm{WE}}) \in [0,u(0)]^N$. Applying Brouwer fixed-point theorem, at least one fixed point $\mathbf{x}^*$ exists and the corresponding WE is just $\mathbf{x}^{\mathrm{WE}} = \mathbf{x}^*$.
\end{proof}

Clearly the delay function $l_i(\cdot)$ in (\ref{equ:delay function}) of the SINR example satisfies the conditions of Theorem \ref{thm:existence of WE}, and thus the existence of WE therein is guaranteed.

\subsection{Uniqueness of Wardrop Equilibrium}
For establishing the uniqueness of WE, we focus on linear environment, i.e., two linear conditions (LC):
\begin{enumerate}
\item[LC1:] congestion function $l_i(\mathbf{x})$ has a linear structure, $l_i(\mathbf{x})=\sum_{j}\tilde{g}_{j,i}x_j$ ($\forall i,j,\;\tilde{g}_{i,j}>0$);
\item[LC2:] the inverse of the user demand function is linear, $u(x)=w-sx$ ($w,s>0$).
\end{enumerate}

We rewrite the WE conditions (\ref{equ:WE1}) and (\ref{equ:WE2}) in a compact form (for compactness of representation, we will use $\mathbf{x}$ instead of $\mathbf{x}^{\mathrm{WE}}$) as

\begin{eqnarray}
\mathbf{x} &\geq& \mathbf{0} \label{equ:constr1}\\
M \mathbf{x}+\mathbf{q} &\geq& \mathbf{0} \label{equ:constr2}\\
\mathbf{x}^T(M\mathbf{x}+\mathbf{q})&=& 0 \label{equ:constr3}
\end{eqnarray}
\noindent where $M=G^T+\mathbf{1}\mathbf{s}^T$, $\mathbf{s}=[s,\cdots,s]^T$,  $\mathbf{w}=[w,\cdots,w]^T$, $\mathbf{p}=[p_1,\cdots,p_N]^T $ and $\mathbf{q}=\mathbf{p}-\mathbf{w}$.
The conditions above form a linear complementarity problem (LCP) and we can utilize some existing results as follows.

\begin{defn}[P-Matrix]
\label{defn:p matrix}
A square matrix $M$ is called a P-matrix if all its principal sub-determinants are strictly positive.
\end{defn}
\begin{lem}[Uniqueness of LCP] \cite[Thm. 4.2]{Murty1972on}
\label{lemma:murty}
The system (10)-(12) has a unique solution for each $\mathbf{q} \in \mathcal{R}^N$ if and only if $M$ is a P-matrix.
\end{lem}

In the case considered in this article, the unique solution of LCP (\ref{equ:constr1})-(\ref{equ:constr3}) corresponds to the unique WE. In many situations, the interferences among the APs are weak, and thus it motivates us to examine if some ``weak interference'' condition can lead to the uniqueness of WE.

\begin{thm}[Uniqueness of WE]
\label{thm:uniqueness of WE}
Under the assumptions of LC1 and LC2, for any given non-negative prices $(p_1,\cdots,p_N)$, if the cross-AP interferences are weak, i.e., $\forall i$, $\sum_{j\neq i}(\tilde{g}_{i,j}+\tilde{g}_{j,i})<2\tilde{g}_{i,i}$, then WE exists and is unique.
\end{thm}
\begin{proof}
From Ger\v{s}gorin's discs theorem (see, e.g., \cite{horn1990matrix}), the condition $\sum_{j\neq i}(\tilde{g}_{i,j}+\tilde{g}_{j,i})<2\tilde{g}_{i,i}$ implies that all the eigenvalues of the symmetric matrix $(G+G^T)/2$ are positive and hence $(G+G^T)/2$ is a positive-definite matrix. Because $\forall \mathbf{x}\in \mathcal{R}^N$, we have $\mathbf{x}^T \mathbf{1}\mathbf{s}^T \mathbf{x} =s(\sum_i x_i)^2\geq 0$. Thus we have $\forall \mathbf{x}\in \mathcal{R}^N$, $\mathbf{x}^T((G+G^T)/2)\mathbf{x}+\mathbf{x}^T(\mathbf{1}\mathbf{s}^T)\mathbf{x} > 0$, i.e., $\mathbf{x}^T (M+M^T)/2 \mathbf{x}=\mathbf{x}^T M \mathbf{x}>0$. Therefore, $M$ is a positive-definite matrix and then a P-matrix. Using Lemma \ref{lemma:murty}, we complete the proof.
\end{proof}

To simplify the analysis, we will only focus on weak cross-AP interferences throughout the remaining analysis of this paper.

\subsection{Characterization of Wardrop Equilibrium for Two APs}
\label{subsect:character}

From the analysis above, we get a clear picture about the number of WEs when the interferences between APs are weak. But when the interferences become strong, how things would be like? In the following, we will analyze a two-AP case to illustrate under assumption LC1 only and try to gain some insights into this problem.

At first, we study how to determine WE given the price vector $\mathbf{p}=(p_1,p_2)$ of SPs. Without loss of generality (WLOG), we assume that $p_1 < p_2$, and the analysis when $p_1 \geq p_2$ is similar. We will divide our analysis into three cases. For convenience, we use substitutions $a_2\triangleq \tilde{g}_{2,1}$, $ b_1\triangleq \tilde{g}_{1,2}$ and $\tilde{g}_{1,1}=\tilde{g}_{2,2}=1$.

\noindent (1) When $\mathbf{p}$ satisfies $p_1,p_2 \geq u(0)$, there is a unique WE $x_1^{\mathrm{WE}}=x_2^{\mathrm{WE}}=0$, which means that the prices set by the SPs are so high that none of users has any incentive to access the network for wireless services.

\noindent (2) When $\mathbf{p}$ satisfies $p_1 < u(0), \; p_2 \geq u(0)$, we have $x_2^{\mathrm{WE}} = 0$ and $x_1^{\mathrm{WE}}$ satisfying  $x_1^{\mathrm{WE}}+p_1 = u(x_1^{\mathrm{WE}})$. Because the left part of the equation, $x_1^{\mathrm{WE}}+p_1$, is monotonically increasing in $x_1^{\mathrm{WE}}$ and the right part of the equation, $u(x_1^{\mathrm{WE}})$, is monotonically decreasing in $x_1^{\mathrm{WE}}$, they must have a single intersection point, and the WE is unique.

\noindent (3) When $\mathbf{p}$ satisfies $0 \leq p_1 < p_2 \leq u(0)$, the situation becomes more involved. Under this condition, $x_1^{\mathrm{WE}}$ and $x_2^{\mathrm{WE}}$ cannot be zero simultaneously (which can be verified by contradiction). Given the market disutility $d$, we define the total demand function as $f(d)\triangleq \sum_{i=1}^{2} x_i^{\mathrm{WE}}(d)$. Thus the intersection point between $u^{-1}(d)$ and $f(d)$ determines the market disutility $d$ at WE. As $u^{-1}(d)$ is a strictly decreasing function in market disutility $d$, while $f(d)$ is not always increasing in $d$ as we will see, there may exist multiple intersection points between $f(d)$ and $u^{-1}(d)$ which correspond to multiple WEs. It is different from the results in \cite{Hayrapetyan2007Network} when cross-AP interferences are not considered. In the sequel, we study the analytical expressions of $f(d)$ under different cases for fixed disutility level.
\begin{enumerate}
    \item When $d$ satisfies $0 \leq d \leq p_1$, we have $x_1^{\mathrm{WE}} = x_2^{\mathrm{WE}} = 0$ (the disutility that end users can stand is smaller than the lowest price, and thus no one would like to access the network), and hence $f(d)=0$;
    \item When $d$ satisfies $p_1<d \leq p_2$, we have $x_2^{\mathrm{WE}} = 0$ and $x_1^{\mathrm{WE}} = d-p_1$ (utilizing equation (\ref{equ:WE1}) and (\ref{equ:WE2})) and hence $f(d)=d-p_1$;
    \item When $d$ satisfies $p_2 < d < u(0)$, there are three cases depending on the strength of cross-AP interferences $(a_2,b_1)$.
        \begin{enumerate}
          \item If $x_1^{\mathrm{WE}}>0$ and $x_2^{\mathrm{WE}}>0$, the conditions of WE are
              \begin{equation}
              \begin{cases}
              p_1+l_1(\mathbf{x}^{\mathrm{WE}})=d; \nonumber\\
              p_2+l_2(\mathbf{x}^{\mathrm{WE}})=d. \nonumber
              \end{cases}
              \end{equation}
              Thus we can derive user distributions between the two APs as
              \begin{equation}
              \begin{cases}
              x_1^{\mathrm{WE}}=\frac{(1-a_2)d+a_2p_2-p_1}{1-a_2b_1}; \nonumber\\
              x_2^{\mathrm{WE}}=\frac{(1-b_1)d+b_1p_1-p_2}{1-a_2b_1}. \nonumber
              \end{cases}
              \end{equation}
              Hence the cross-AP interferences region $\{(a_2,b_1)|a_2<\frac{d-p_1}{d-p_2},b_1<\frac{d-p_2}{d-p_1} \;\mathrm{or}\;a_2>\frac{d-p_1}{d-p_2},b_1>\frac{d-p_2}{d-p_1}\}$ can lead to $x_1^{\mathrm{WE}}>0$ and $x_2^{\mathrm{WE}}>0$. The total demand function is
              \begin{equation}
              f(d)= \frac{(2-a_2-b_1)d+(a_2-1)p_2+(b_1-1)p_1}{1-a_2b_1}. \nonumber
              \end{equation}
          \item Similarly, if $x_1^{\mathrm{WE}}>0$ and $x_2^{\mathrm{WE}}=0$, then the conditions of WE become
              \begin{equation}
              \begin{cases}
              p_1+l_1(\mathbf{x}^{\mathrm{WE}})=d; \nonumber\\
              p_2+l_2(\mathbf{x}^{\mathrm{WE}})\geq d. \nonumber
              \end{cases}
              \end{equation}
              and we have $x_1^{\mathrm{WE}}=d-p_1$.
              Cross-AP interferences region should satisfy $\{(a_2,b_1)|b_1\geq \frac{d-p_2}{d-p_1} \}$ and the total demand function is $f(d)=d-p_1$.
          \item If $x_1^{\mathrm{WE}}=0$ and $x_2^{\mathrm{WE}}>0$, we have $x_2^{\mathrm{WE}}=d-p_2$ and cross-AP interferences region should satisfy $\{(a_2,b_1)|a_2\geq \frac{d-p_1}{d-p_2}\}$. The total demand function is $f(d)=d-p_2$.
        \end{enumerate}

         To sum up, we divide cross-AP interferences region into four sub-regions (see Figure (\ref{fig:division1})):
        \begin{enumerate}
            \item when $(a_2,b_1)$ falls in region (\Rmnum{1}), we have $x_1^{\mathrm{WE}}>0$, $x_2^{\mathrm{WE}}>0$ and $f(d)=\frac{(2-a_2-b_1)d+(a_2-1)p_2+(b_1-1)p_1}{1-a_2b_1}$;
            \item when $(a_2,b_1)$ falls in region (\Rmnum{2}), we have $x_1^{\mathrm{WE}}=0$, $x_2^{\mathrm{WE}}>0$ and $f(d)=d-p_2$;
            \item when $(a_2,b_1)$ falls in region (\Rmnum{3}), three cases exist:
                \begin{itemize}
                    \item $x_1^{\mathrm{WE}}>0$, $x_2^{\mathrm{WE}}>0$ and $f(d)=\frac{(2-a_2-b_1)d+(a_2-1)p_2+(b_1-1)p_1}{1-a_2b_1}$,
                    \item or $x_1^{\mathrm{WE}}=0$, $x_2^{\mathrm{WE}}>0$ and $f(d)=d-p_2$,
                    \item or $x_1^{\mathrm{WE}}>0$, $x_2^{\mathrm{WE}}=0$ and $f(d)=d-p_1$;
                \end{itemize}
            \item when $(a_2,b_1)$ falls in region (\Rmnum{4}), we have $x_1^{\mathrm{WE}}>0$, $x_2^{\mathrm{WE}}=0$ and $f(d)=d-p_1$.
        \end{enumerate}
\end{enumerate}
\noindent However, when market disutility $d$ changes, the region that $(a_2,b_1)$ belongs to in Figure (\ref{fig:division1}) also changes, we need to have the region be further divided and the division should be independent of market disutility $d$. Because the intersection point of the four sub-regions in Figure (\ref{fig:division1}) is $\left(\frac{d-p_1}{d-p_2},\frac{d-p_2}{d-p_1}\right)$, when $d$ increases from $p_2$ to $u(0)$, the intersection point moves along the trajectory of $a_2b_1=1$, which enables us to further divide the cross-AP interferences region into five sub-regions (see Figure (\ref{fig:division2})). Due to aforementioned analysis, we are able to get a detailed characterization of the total demand function $f(d)$ as follows:
\begin{enumerate}
    \item when $(a_2,b_1)$ falls in region (a) in Figure (\ref{fig:division2}), i.e., $\{(a_2,b_1)| 0\leq a_2<\frac{u(0)-p_1}{u(0)-p_2}, \; 0\leq b_1<\frac{u(0)-p_2}{u(0)-p_1}\}$, as $d$ increases from $p_2$ to $u(0)$, the region that $(a_2,b_1)$ belongs to in Figure (\ref{fig:division1}) moves from region (\Rmnum{4}) to region (\Rmnum{1}), and thus we have
        \begin{equation}
        f(d)=
        \begin{cases}
        d-p_1, \; \mathrm{when}\; p_2<d \leq \frac{p_2-b_1p_1}{1-b_1}; \\
        \frac{(2-a_2-b_1)d+(a_2-1)p_2+(b_1-1)p_1}{1-a_2b_1}, \; \\
        \;\;\;\;\;\;\;\;\;\; \mathrm{when}\; \frac{p_2-b_1p_1}{1-b_1}< d < u(0);
        \end{cases}
        \end{equation}
    \item similarly, when $(a_2,b_1)$ falls in region (b), i.e., $\{(a_2,b_1)|a_2\geq \frac{u(0)-p_1}{u(0)-p_2}, \; 0\leq b_1\leq\frac{u(0)-p_2}{u(0)-p_1}, \;a_2b_1<1\}$, as $d$ increases from $p_2$ to $u(0)$, it moves from region (\Rmnum{4}) to region (\Rmnum{1}) and finally moves to region (\Rmnum{2}), and thus we have
        \begin{equation}
              f(d)=
              \begin{cases}
              d-p_1, \; \mathrm{when}\; p_2<d \leq \frac{p_2-b_1p_1}{1-b_1}; \\
              \frac{(2-a_2-b_1)d+(a_2-1)p_2+(b_1-1)p_1}{1-a_2b_1}, \; \\
              \;\;\;\;\;\;\;\;\;\; \mathrm{when}\; \frac{p_2-b_1p_1}{1-b_1}< d < \frac{p_1-a_2p_2}{1-a_2};\\
              d-p_2, \; \mathrm{when}\; \frac{p_1-a_2p_2}{1-a_2}\leq d <u(0);
              \end{cases}
        \end{equation}
    \item when $(a_2,b_1)$ falls in region (c), i.e., $\{(a_2,b_1)|a_2\geq \frac{u(0)-p_1}{u(0)-p_2}, \; 0\leq b_1\leq\frac{u(0)-p_2}{u(0)-p_1}, \;a_2b_1>1\}$, it moves from region (\Rmnum{4}) to region (\Rmnum{3}) and finally moves to region (\Rmnum{2}), and thus we have
        \begin{equation}
              f(d)=
              \begin{cases}
              d-p_1, \; \mathrm{when}\; p_2<d \leq \frac{p_1-a_2p_2}{1-a_2}; \\
              \frac{(2-a_2-b_1)d+(a_2-1)p_2+(b_1-1)p_1}{1-a_2b_1} \; \\
              \mathrm{or} \; d-p_1 \\
              \mathrm{or} \; d-p_2, \\
              \;\;\;\;\;\;\;\;\;\; \mathrm{when}\; \frac{p_1-a_2p_2}{1-a_2}< d < \frac{p_2-b_1p_1}{1-b_1};\\
              d-p_2, \; \mathrm{when}\; \frac{p_2-b_1p_1}{1-b_1}\leq d <u(0);
              \end{cases}
        \end{equation}
    \item when $(a_2,b_1)$ falls in region (d), i.e., $\{(a_2,b_1)|a_2 > \frac{u(0)-p_1}{u(0)-p_2}, \; b_1 > \frac{u(0)-p_2}{u(0)-p_1}\}$, it moves from region (\Rmnum{4}) to region (\Rmnum{3}), and thus we have
        \begin{equation}
              f(d)=
              \begin{cases}
              d-p_1, \; \mathrm{when}\; p_2<d \leq \frac{p_1-a_2p_2}{1-a_2}; \\
              \frac{(2-a_2-b_1)d+(a_2-1)p_2+(b_1-1)p_1}{1-a_2b_1} \; \\
              \mathrm{or} \; d-p_1 \\
              \mathrm{or} \; d-p_2, \\
              \;\;\;\;\;\;\;\;\;\; \mathrm{when}\; \frac{p_1-a_2p_2}{1-a_2}< d < u(0)
              \end{cases}
        \end{equation}
    \item when $(a_2,b_1)$ falls in region (e), i.e., $\{(a_2,b_1)|0\leq a_2<\frac{u(0)-p_1}{u(0)-p_2}, \; b_1>\frac{u(0)-p_2}{u(0)-p_1}\}$, as $d$ increases from $p_2$ to $u(0)$, $(a_2,b_1)$ always stays in region (\Rmnum{4}), and thus we have
        \begin{equation}
              f(d)=d-p_1,\; \mathrm{when} \; p_2 < d < u(0).
        \end{equation}
 \end{enumerate}

\begin{figure}[!htb]
\center
\includegraphics[width=2.2in]{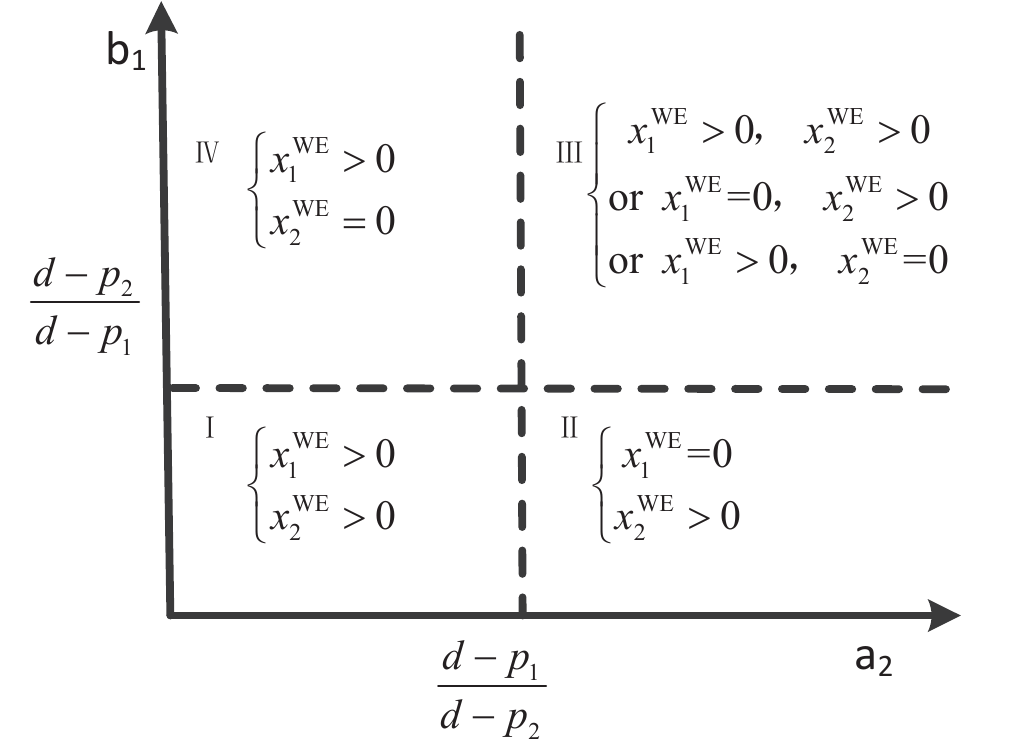}\\
\caption{Cross-AP interferences sub-regions given market disutility d}
\label{fig:division1}
\end{figure}

\begin{figure}[!htb]
\center
\includegraphics[width=2.8in]{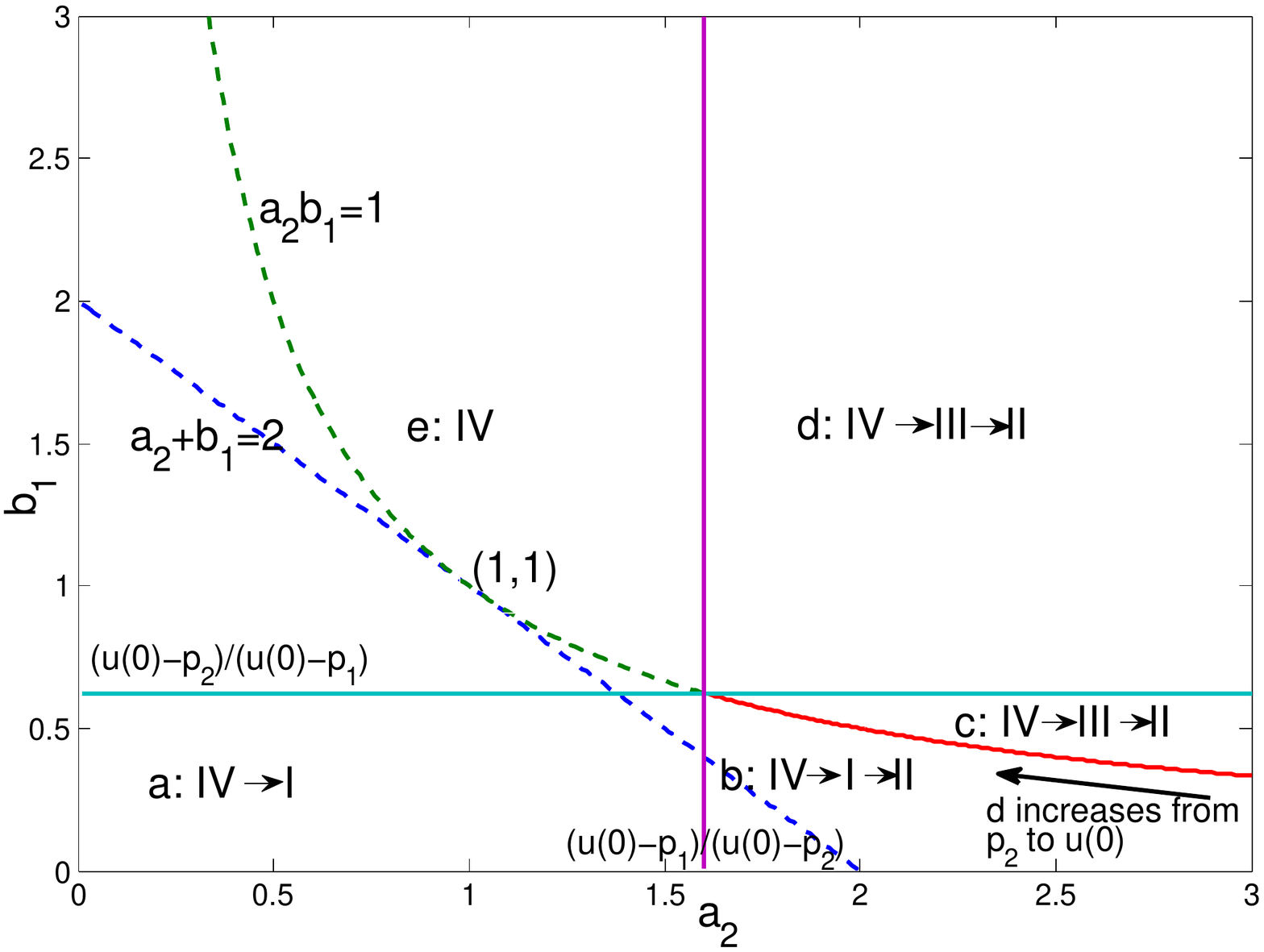}\\
\caption{General cross-AP interferences sub-regions}
\label{fig:division2}
\end{figure}

Take part of region (b) for an illustration, see Figure (\ref{fig:fd2}). When $(a_2,b_1)\in \{(a_2,b_1)|a_2b_1<1,a_2+b_1>2\}\bigcap \mathrm{region} \;(\mathrm{b})$, the total demand function $f(d)$ decreases in $d$ when $d$ falls in the interval $(\frac{p_2-b_1p_1}{1-b_1},\frac{p_1-a_2p_2}{1-a_2})$. It is a very counter-intuitive result, since the increase of end users' tolerance of the market disutility can lead to the decrease of the number of end users. There may exist multiple WEs due to the following two reasons:
\begin{enumerate}
\item   the interference regions involve region \Rmnum{3} when disutility $d$ varies (three cases can happen);
\item   the interferences regions involve region \Rmnum{1} and $(a_2,b_1)$ falls in the region $\{(a_2,b_1)|a_2b_1<1,a_2+b_1>2,a_2>0, b_1>0\}$ (this is due to non-monotonicity of the total flow).
\end{enumerate}
\noindent Therefore, when we change the prices $(p_1,p_2)$, as long as $a_2+b_1<2$, there exists a unique WE for all market demand functions. This result coincides with Theorem \ref{thm:uniqueness of WE}, while assumption LC2 is not required here.

\begin{figure}[!htb]
\center
\includegraphics[width=2.2in]{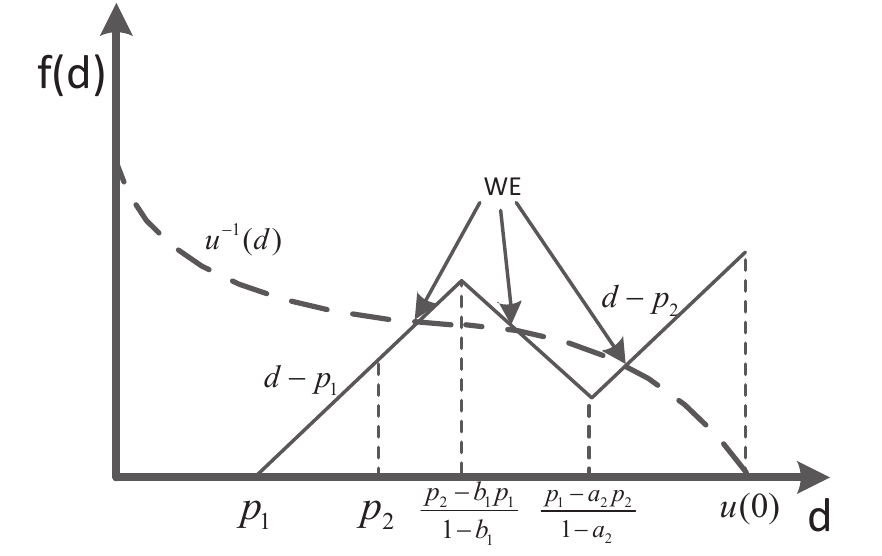}\\
\caption{An illustration when $(a_2,b_1)$ falls in region (b)}
\label{fig:fd2}
\end{figure}

\section{\textsc{Oligopoly Analysis}}
\label{sect:oligopoly analysis}

Generally, the oligopoly analysis of this problem is complex and intractable. In this section, in order to obtain analytical results and gain insights, we focus on a two-AP case under the linear assumptions (LC1, LC2) with weak cross-AP interferences. Through the following analysis, we characterize the APs' prices and user distributions at both monopoly and duopoly equilibria.

\subsection{Monopoly Case}
\label{subsect:monopoly case}

When the APs are controlled by one SP, the goal of the SP is to maximize the total profits of the APs. We prove that when there are only two APs and the interferences between them are weak, the pricing strategy can be expressed analytically.

Price differentiation (PD) is a pricing strategy in which an SP offers the same service or product at different prices in different markets. PD can be a feature of monopolistic and oligopolistic markets. By means of PD, SP can absorb consumer's surplus from the consumer to raise its revenue. If PD is permitted, the problem can be formulated as

\begin{eqnarray}
\max_{\mathbf{p}} & \mathbf{p}^T\mathbf{x} \label{equ:opt1} \\
\mathrm{s.t.} \; & \mathrm{LCP}\; (\ref{equ:constr1})-(\ref{equ:constr3}). \nonumber
\end{eqnarray}

\begin{thm}[Monopoly Equilibrium with PD]
\label{thm:monopoly equilibrium with PD}
Under the assumptions of LC1 and LC2, when cross-AP interferences are weak, i.e., $a_2+b_1<2$, and the prices are differentiated, ME exists and is unique. The corresponding prices are:
\begin{equation}
\mathbf{p}^{\mathrm{ME}} = (M^{-1}+(M^{-1})^T)^{-1}M^{-1}\mathbf{w}
\end{equation}
Meanwhile the user distributions between the APs are:
\begin{equation}
\mathbf{x}^{\mathrm{ME}} = M^{-1}(I-(M^{-1}+(M^{-1})^T)^{-1}M^{-1})\mathbf{w}
\end{equation}
\end{thm}

\begin{proof}
Firstly, we prove by contradiction that both APs have positive demands, i.e., $x_1^{\mathrm{WE}}(\mathbf{p}^{\mathrm{ME}})>0$ and $x_2^{\mathrm{WE}}(\mathbf{p}^{\mathrm{ME}})>0$. WLOG, suppose that at the ME, we have $x_1^{\mathrm{WE}}=0$ and $x_2^{\mathrm{WE}}>0$, the corresponding disutility is $d^{*}$. From the WE constraints $d^{*}=w-sx_2^{\mathrm{WE}}$, we get $x_2^{\mathrm{WE}}=(w-d^{*})/s$. Thus the cost due to delay is $(x_2^{\mathrm{WE}})^2=((w-d^{*})/s)^2$. Under the same degree of the disutility $d^{*}$, if there exists $p_1$ and $p_2$ that can make $x_1^{\mathrm{WE}}>0 \; x_2^{\mathrm{WE}}>0$, then the congestion cost incurred is
\begin{eqnarray}
& &(x_1^{\mathrm{WE}}+a_2x_2^{\mathrm{WE}})x_1^{\mathrm{WE}}+(x_2^{\mathrm{WE}}+b_1x_1^{\mathrm{WE}})x_2^{\mathrm{WE}}\\
&=& (x_1^{\mathrm{WE}}+x_2^{\mathrm{WE}})^2-(2-a_2-b_1)x_1^{\mathrm{WE}}x_2^{\mathrm{WE}}\\
&=& ((w-d^{*})/s)^2-(2-a_2-b_1)x_1^{\mathrm{WE}}x_2^{\mathrm{WE}} \\
&<& ((w-d^{*})/s)^2,
\end{eqnarray}
which means that SP can get more profits if both of the APs are at work. From Figure (\ref{fig:division2}), when $d^{*}$ is fixed, as long as $a_2+b_1<2$, we can tune the prices $p_1\; p_2$ that make $(a_2,b_1)$ falls in region (\Rmnum{1}). Thus both APs have positive demands and the constraints of the problem become $M\mathbf{x}+\mathbf{q}=0$. It is a non-constrained convex optimization problem and can be easily solved.
\end{proof}

Otherwise, without PD, the problem is to solve
\begin{eqnarray}
\max_{p} & p\mathbf{1}^T\mathbf{x} \label{equ:opt2} \\
\mathrm{s.t.} \; & \mathrm{LCP}\; (\ref{equ:constr1})-(\ref{equ:constr3}), \nonumber
\end{eqnarray}
\noindent and we have the following result.

\begin{thm}[Monopoly Equilibrium without PD]
\label{thm:monopoly equilibrium without PD}
Under the assumptions of LC1 and LC2, when cross-AP interferences are weak, i.e., $a_2+b_1<2$, and the prices are not differentiated, ME exists and is unique. The corresponding prices are:
\begin{equation}
p^{\mathrm{ME}} = w/2.
\end{equation}
Meanwhile the user distributions between the APs are:
\begin{equation}
\mathbf{x}^{\mathrm{ME}} = M^{-1}\mathbf{w}/2.
\end{equation}
\end{thm}

\begin{proof}
When $p_1=p_2$ is considered, it is a special case of Figure(\ref{fig:division1}). The region $\{(a_2,b_1)|0\leq a_2,b_1<1\}$ corresponds both of the APs have positive flow, the region $\{(a_2,b_1)|a_2\geq 1,b_1\leq 1,a_2+b_1<2\}$ corresponds to $x_2^{\mathrm{WE}}>0$ $x_1^{\mathrm{WE}}=0$ and the region $\{(a_2,b_1)|a_2\leq 1,b_1\geq 1,a_2+b_1<2\}$ corresponds to $x_1^{\mathrm{WE}}>0$ $x_2^{\mathrm{WE}}=0$. We prove by classified discussions.

\noindent(1) When $a_2\leq 1,b_1\leq 1$, since both APs have positive flows, the constraints of WE reduce to $M\mathbf{x}+\mathbf{q}=0$. Thus the profit of SP is $p\mathbf{1}^T(M^{-1}(\mathbf{w}-p\mathbf{1}))$ and SP can get its optimal profit by setting the price as $p^{\mathrm{WE}}=w/2$.

\noindent(2) When $a_2\geq 1,b_1\leq 1,a_2+b_1<2$, we have $x_2^{\mathrm{WE}}>0$ $x_1^{\mathrm{WE}}=0$. Thus the profit of SP is $px_2^{\mathrm{WE}}$. Since $w-sx_2^{\mathrm{WE}}=p+x_2^{\mathrm{WE}}$, the profit can be expressed as $p(w-p)/(s+1)$ and SP can also get its optimal profit by setting the price as $p^{\mathrm{WE}}=w/2$. The discussion of the other case is similar.
Thus we complete the proof.
\end{proof}
This result is interesting because when the prices are not differentiated, the SP simply sets $p^{\mathrm{ME}}=w/2$ regardless of the interferences between the APs.

\subsection{Duopoly Case}
\label{subsect:duopoly case}
When the APs are controlled by different SPs, the goal of each SP is to maximize his own profits. Similar to the monopoly case, we prove that each SP's profit-maximizing problem can be solved analytically when there only exist two APs.

For each SP, its optimization problem given the price of the other SP $p_{-i}$ can be expressed as

\begin{eqnarray}
\label{equ:DE problem}
\max_{p_i} & p_ix_i \label{equ:opt3} \\
\mathrm{s.t.} \; & \mathrm{LCP}\; (\ref{equ:constr1})-(\ref{equ:constr3}). \nonumber
\end{eqnarray}

The following theorem characterizes SPs' pricing strategies.

\begin{thm}[Duopoly Equilibrium]
\label{thm:duopoly equilibrium}
Under the assumptions of LC1 and LC2, when the cross-AP interferences are weak, i.e., $0\leq a_2+b_1<2$, for each $\mathrm{SP}_i$, the best pricing strategy is dependent on the cross-AP interference stated as follows.
\begin{enumerate}
\item When the cross-AP interferences satisfy $0\leq a_2 \leq 1$ and $0\leq a_2 \leq 1$, DE exists and is unique, and the pricing strategies are:
        \begin{equation}
        \begin{cases}
        p_1^{*}  = \frac{w[(a_2+s)(1-b_1)+2(1-a_2)(1+s)]}{4(1+s)^2-(s+a_2)(s+b_1)}, \\
        p_2^{*}  = \frac{w[(b_1+s)(1-a_2)+2(1-b_1)(1+s)]}{4(1+s)^2-(s+a_2)(s+b_1)};
        \end{cases}
        \end{equation}
        meanwhile the user distributions between the APs are:
        \begin{equation}
        \begin{cases}
        x_1^{*} = \frac{(w-p_1^{*})(1+s)-(w-p_2^{*})(a_2+s)}{(1+s)^2-(s+a_2)(s+b_1)}, \\
        x_2^{*} = \frac{(w-p_2^{*})(1+s)-(w-p_1^{*})(b_1+s)}{(1+s)^2-(s+a_2)(s+b_1)}.
        \end{cases}
        \end{equation}
\item When the cross-AP interferences satisfy $a_2>1$, $0\leq b_1 <1$ and $a_2+b_2<2$, the pricing strategies are:
        \begin{equation}
        \begin{cases}
        \mathrm{if}\; p_2^{*}x_2^{*}>\frac{(a_2-1)w^2}{(s+a_2)^2}, \; p_2^{\mathrm{DE}}=p_2^{*}, \; p_1^{\mathrm{DE}}=p_1^{*}, \; \\
        \mathrm{if}\; p_2^{*}x_2^{*}\leq \frac{(a_2-1)w^2}{(s+a_2)^2}, \; p_2^{\mathrm{DE}}=\frac{(a_2-1)w}{s+a_2}, \; p_1^{\mathrm{DE}}\geq 0. \;
        \end{cases}\nonumber
        \end{equation}
\item When the cross-AP interferences satisfy $b_1>1$, $0\leq a_2 <1$ and $a_2+b_2<2$, the pricing strategies are:
        \begin{equation}
        \begin{cases}
        \mathrm{if}\; p_1^{*}x_1^{*}>\frac{(b_1-1)w^2}{(s+b_1)^2}, \; p_1^{\mathrm{DE}}=p_1^{*}, \; p_2^{\mathrm{DE}}=p_2^{*}, \; \\
        \mathrm{if}\; p_1^{*}x_1^{*}\leq \frac{(b_1-1)w^2}{(s+b_1)^2}, \; p_1^{\mathrm{DE}}=\frac{(b_1-1)w}{s+b_1}, \; p_2^{\mathrm{DE}}\geq 0. \;
        \end{cases}\nonumber
        \end{equation}
\end{enumerate}
\end{thm}

\begin{proof}
\\
\indent 1) When the cross-AP interferences satisfy $0\leq a_2 \leq 1$ and $0\leq a_2 \leq 1$, WLOG, we assume $x_1=0$ and $x_2>0$, and thus we have $d_1\geq d_2$ (WE condition), i.e., $p_1+a_2x_2\geq p_2+x_2$. $\mathrm{SP}_1$ can always lower the price $p_1$ it charges to gain positive profits. Thus both SPs can get positive profits and flow. Solving problem (\ref{equ:DE problem}), the best-response function of $\mathrm{SP}_1$ is :
    \begin{equation}
    \mathrm{BR}_1(p_2)=\max\left\{ \min\left\{\frac{(a_2+s)p_2+w(1-a_2)}{2(1+s)},\; u(0) \right\},\; 0 \right\} \nonumber
    \end{equation}
    and that of $\mathrm{SP}_2$ is:
    \begin{equation}
    \mathrm{BR}_2(p_1)=\max\left\{ \min\left\{\frac{(b_1+s)p_1+w(1-b_1)}{2(1+s)},\; u(0) \right\},\; 0 \right\} \nonumber
    \end{equation}
    Because of $x_1>0$ and $x_2>0$, DE cannot be obtained at the boundaries of both best-response functions, i.e., $\forall i \in \{1,2\}$, $p_i\neq 0,u(0)$. Thus DE is determined by the intersection point of the two best-response functions' linear part and we get the prices $\mathbf{p}^{\mathrm{DE}}$ and user distributions between SPs $\mathbf{x}^{\mathrm{WE}}(\mathbf{p}^{\mathrm{DE}})$ at DE. \\
\indent 2) When the cross-AP interferences satisfy $a_2>1$, $0\leq b_1 <1$ and $a_2+b_2<2$, $\mathrm{SP}_2$ has the power to expel $\mathrm{SP}_1$ from the market. When $\mathrm{SP}_1$ is expelled from market, we have $x_1=0$ and $x_2>0$, as the analysis above, the price of $\mathrm{SP}_2$ should satisfy $p_2\leq (a_2-1)x_2$ to guarantee that no matter how much $\mathrm{SP}_1$ charges, no users would choose it for service. It follows that $d_2=w-sx_2=p_2+x_2$, hence the profit of $\mathrm{SP}_2$ is $p_2(w-p_2)/(s+1)$. If there is no constraint of $p_2$, then $\mathrm{SP}_2$ can get its optimal profit by setting $p_2=w/2$. However, as we have the constraint $p_2\leq (a_2-1)x_2$ which implies that $p_2\leq (a_2-1)w/(s+a_2)<w/2$, $\mathrm{SP}_2$ have to set the price as $p_2=(a_2-1)w/(s+a_2)$ to get the maximum profit $(a_2-1)w^2/(s+a_2)^2$. When the two SPs coexists, the profit of $\mathrm{SP}_2$ is $p_2^{*}x_2^{*}$ and $\mathrm{SP}_2$ would choose the price that achieves the optimal profit.\\
\indent 3) When the cross-AP interferences satisfy $b_1>1$, $0\leq a_2 <1$ and $a_2+b_2<2$, the analysis is similar to 2).
\end{proof}

\section{\textsc{Comparison of Monopoly and Duopoly}}
\label{sect:comparison of monopoly and duopoly}
In this section, following the assumptions of Section \ref{sect:oligopoly analysis}, we further specialize to the case where cross-AP interferences are symmetric, i.e., $0\leq a_2=b_1<1$.

\subsection{Price of Competition on Social Welfare}
\label{subsect:price of competition on social welfare}

In order to compare the impacts of monopoly and duopoly on social welfare, we define the price of competition on social welfare (PoCS) as:
\begin{equation}
\mathrm{PoCS} = \frac{\mathrm{SW}(\mathbf{x}^{\mathrm{WE}}(\mathbf{p}^{\mathrm{ME}}))}{\mathrm{SW}(\mathbf{x}^{\mathrm{WE}}(\mathbf{p}^{\mathrm{DE}}))}
\end{equation}

We study the properties of PoCS and the results are summarized in the following theorem.

\begin{thm}[Bound on PoCS]
Under the assumptions of LC1 and LC2 with weak and symmetric cross-AP interferences, a lower bound of PoCS is $\frac{3}{4}$ and this bound is tight in that it is attained when the elasticity of market demand function, i.e., $\frac{1}{s}$ , approaches zero. On the other hand, PoCS is unbounded from above.
\end{thm}
\begin{proof}
Using the results of Theorem \ref{thm:monopoly equilibrium with PD} and Theorem \ref{thm:duopoly equilibrium}, we obtain the analytical expression of PoCS as
\begin{equation}
\mathrm{PoCS} = \frac{(3s+1+a_2)(s+2-a_2)^2}{4(s+1)(s^2+3s+1-2a_2s-a_2^2)}
\end{equation}
Then it is easy to verify that $\mathrm{PoCS}-\frac{3}{4} \geq 0$ and $\lim_{s \rightarrow \infty}\mathrm{PoCS}=\frac{3}{4}$. When the elasticity of market demand $\frac{1}{s}$ approaches to infinity and cross-AP interferences approach to $1$, we have $\lim_{s\rightarrow 0^{+},a_2\rightarrow 1^{-}}\mathrm{PoCS}(a_2,s)=\infty$, which means that PoCS is unbounded from above.
\end{proof}

\subsection{Price of Competition on Profits}
\label{subsect:price of competition on profits}

In order to compare the impacts of monopoly and duopoly on SPs' profits, we define the price of competition on profits (PoCP) as
\begin{equation}
\mathrm{PoCP} = \frac{\sum_{i}p_i^{\mathrm{ME}}x_i^{\mathrm{WE}}(\mathbf{p}^{\mathrm{ME}})}
{\sum_{i}p_i^{\mathrm{DE}}x_i^{\mathrm{WE}}(\mathbf{p}^{\mathrm{DE}})}.
\end{equation}
Using the results of Theorem \ref{thm:monopoly equilibrium with PD} and \ref{thm:duopoly equilibrium}, we obtain

\begin{equation}
\mathrm{PoCP} = \frac{(s-a_2+2)^2}{4(1-a_2)(s+1)}
\end{equation}
We can then verify that $\mathrm{PoCP}-1\geq 0$ and $\lim_{s,a_2 \rightarrow 0}\mathrm{PoCP}=1$, so we have the following theorem.

\begin{thm}[Lower bound of PoCP]
Under the assumptions of LC1 and LC2 with weak and symmetric cross-AP interferences, a lower bound of PoCP is $1$ and this bound is tight in that is attained when the elasticity of market demand function, i.e. $\frac{1}{s}$ , diverges and cross-AP interferences approach zero. On the other hand, PoCP is unbounded from above.
\end{thm}

\section{Numerical Illustrations}
\label{sect:numerical}
In this section, we perform some numerical experiments to illustrate the analysis in Section \ref{sect:oligopoly analysis} and Section \ref{sect:comparison of monopoly and duopoly}.

Firstly, we study the relationship between the equilibrium prices and cross-AP interferences. We set $w=s=1$, $b_1=0.3$ and vary $a_2$ from $0$ to $1.7$, then we get the variations of prices, user distributions and the profits of APs at ME and DE, respectively (see Figure (\ref{fig:equilibrium})). From the figure, we can see that the prices at ME are always higher than those of DE which accords with our intuitive ideas. As $a_2$ increases, SP increases $\mathrm{AP}_2$'s access price to control the flow amount, thus limiting the cross-AP interference from $\mathrm{AP}_2$ to $\mathrm{AP}_1$; while at DE, $\mathrm{SP}_2$ even continues to reduce $\mathrm{AP}_2$'s price, which aggravates $\mathrm{SP}_1$'s situation. Interestingly, at ME, the monopolistic SP prefers to equalize the user distributions between two APs.

\begin{figure}[!ht]
\centering
\subfigure[AP prices set by SPs at ME and DE]{
\includegraphics[width=2in]{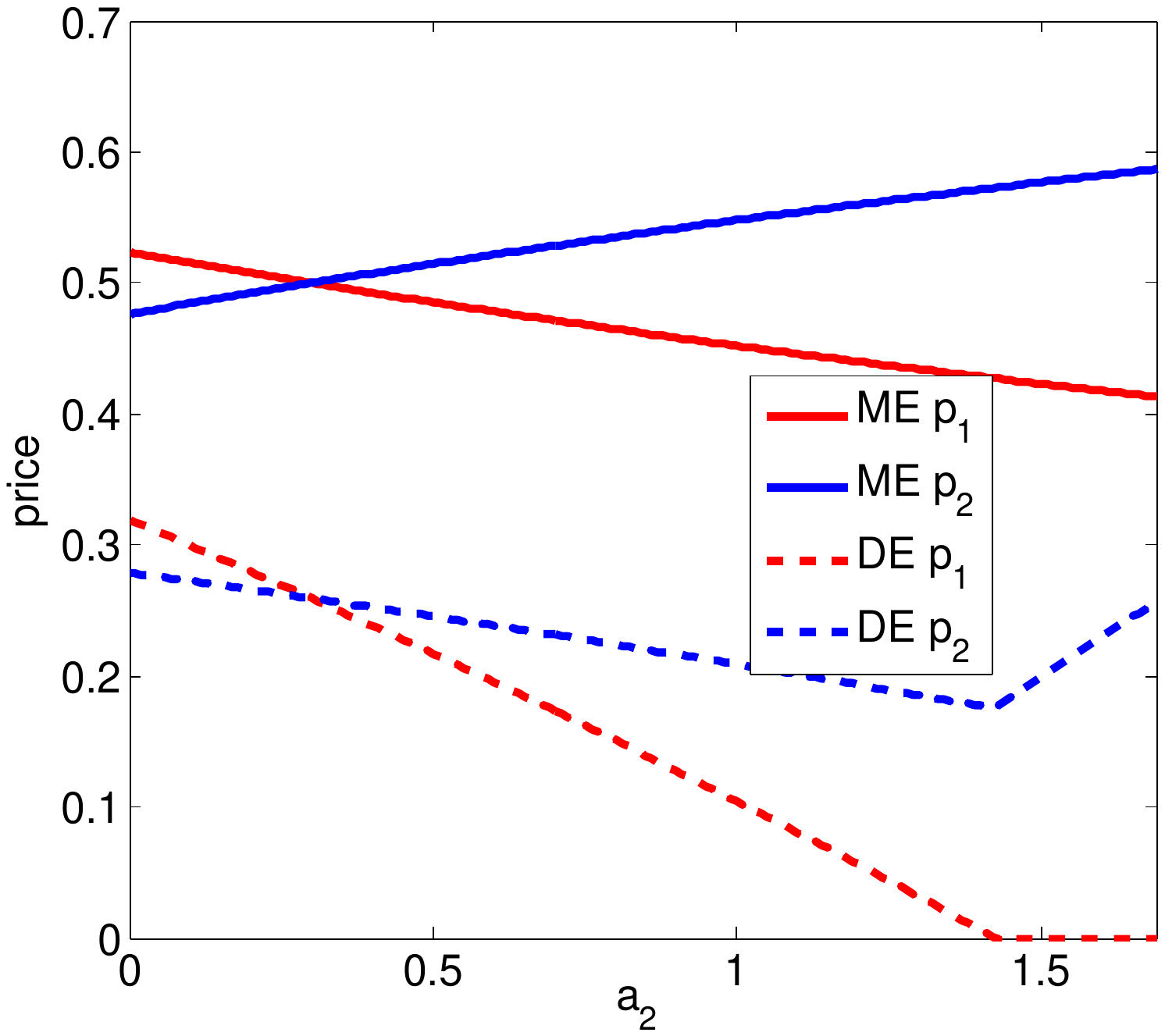}}
\subfigure[User distributions between APs at ME and DE]{
\includegraphics[width=2in]{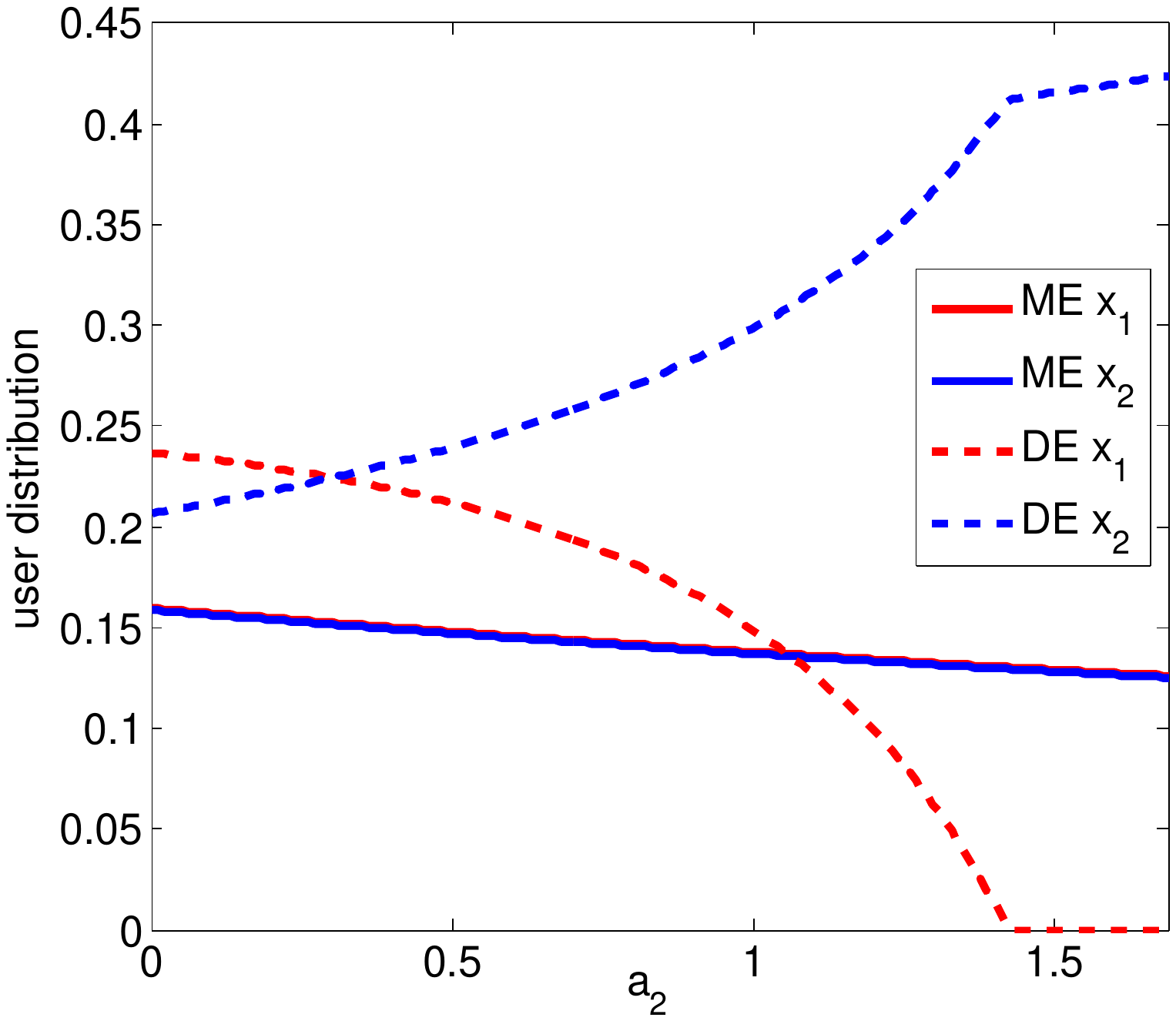}}
\subfigure[Profit of each AP at ME and DE]{
\includegraphics[width=2in]{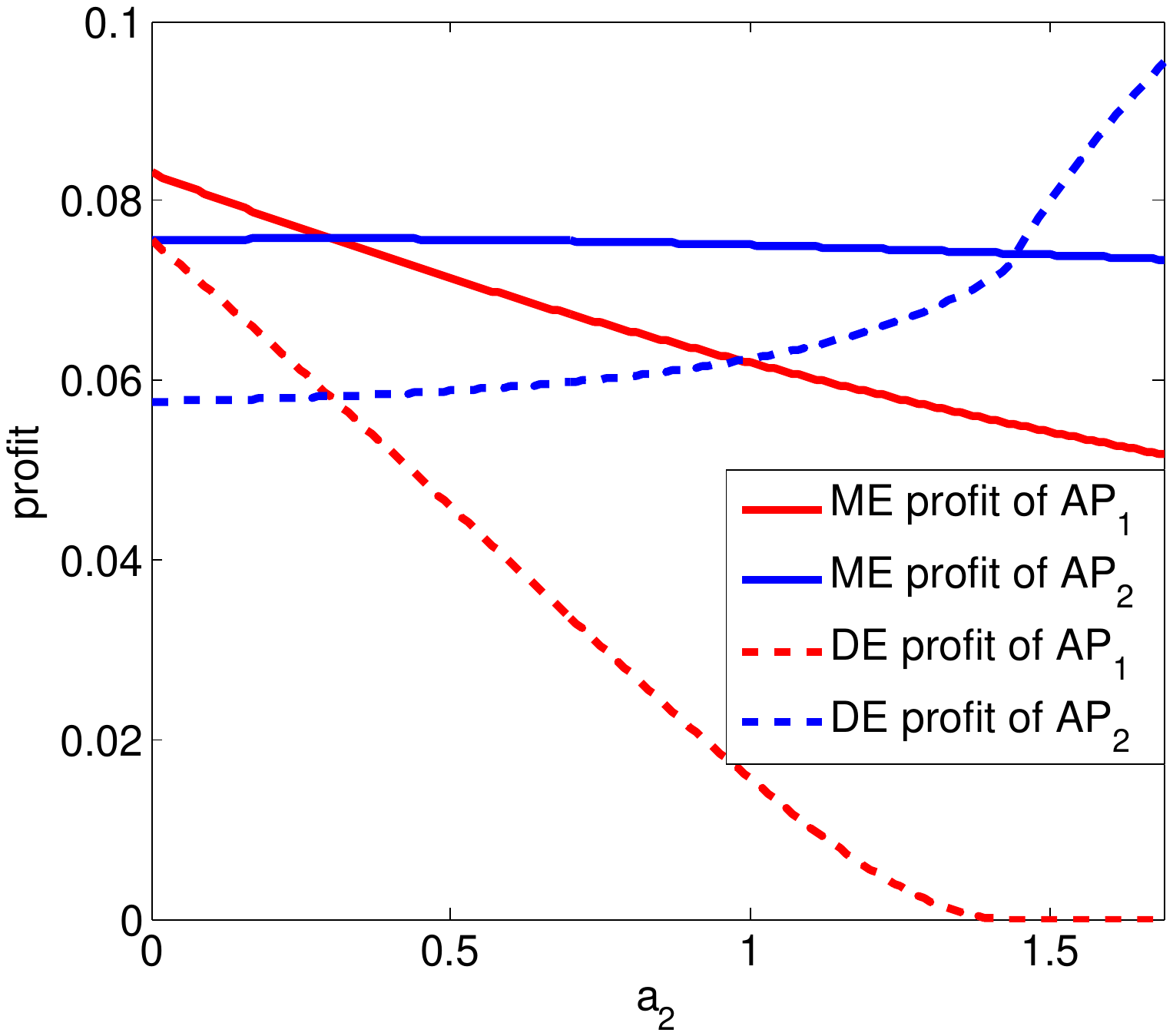}}
\caption{AP prices, user distributions and profits at equilibria}
\label{fig:equilibrium}
\end{figure}

Then we study PoCS and PoCP for both symmetric and asymmetric cases. We set $s$ equal to $0.1$, $1$ and $10$, respectively, vary $a_2$ from $0$ to $1$ and get the PoCS and PoCP variation tendencies (see Figure (\ref{fig:PoCS PoCP symmetric})). From Figure (\ref{fig:PoCS PoCP symmetric}a), we find that when the elasticity of market demand equals $1$ and the symmetric cross-AP interferences are less than $0.7$, duopoly is more favorable than monopoly from the point of view of social welfare. However, when the symmetric cross-interferences are strong, monopoly is more favorable than monopoly from the point of view of both social welfare and SPs' profits.

\begin{figure}[!ht]
\centering
\subfigure[PoCS under symmetric cross-AP interferences]{
\includegraphics[width=2in]{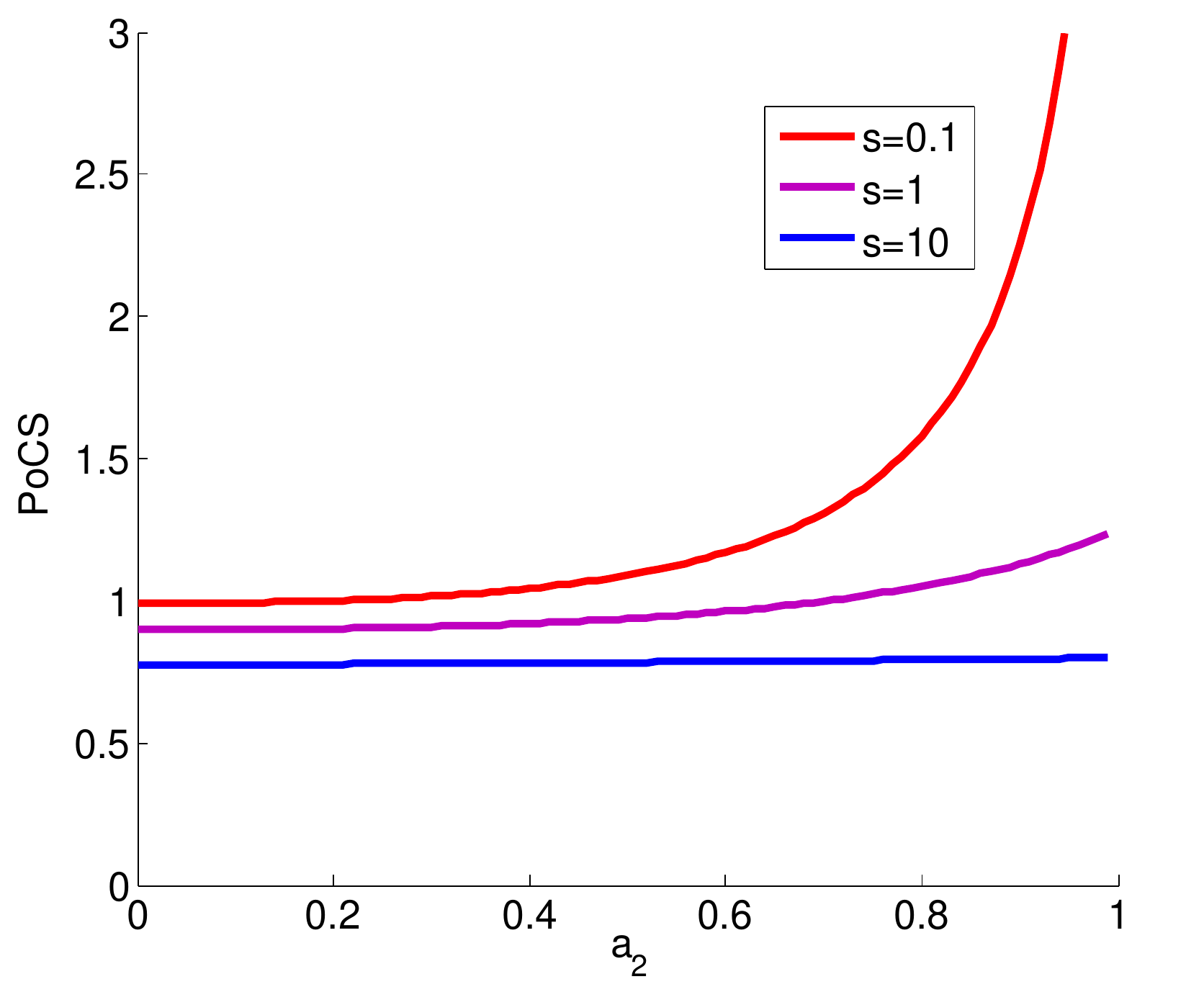}} \label{fig:PoCS symmetric}
\subfigure[PoCP under symmetric cross-AP interferences]{
\includegraphics[width=2in]{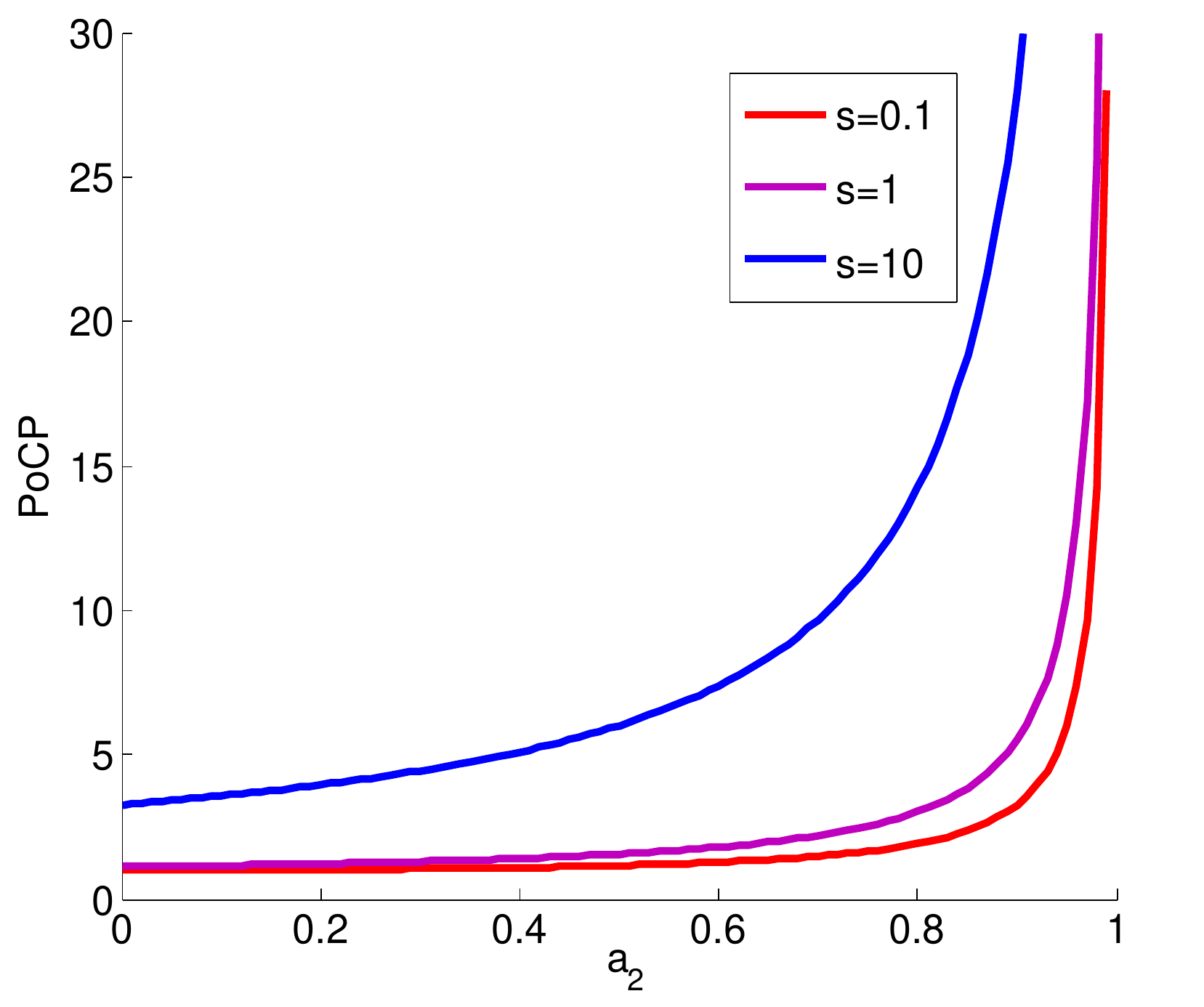}} \label{fig:PoCP symmetric}
\caption{PoCS and PoCP under symmetric cross-AP interferences}
\label{fig:PoCS PoCP symmetric}
\end{figure}

\section{\textsc{Conclusion}}
\label{sect:conclusion}
Through the analysis in this paper, we find that when there exist cross-AP interferences in the system, the performance of monopoly is guaranteed from both the perspectives of social welfare and profits. Indeed, as cross-AP interferences become strong, the penalty due to competition on both social welfare and SPs' profits become large and even unbounded. Such results shed new light on understanding the role of competition and efficiency in wireless SP markets, where cross-AP interferences are inevitable.

%In future work, generalization includes characterization of equilibria in general congestion and demand models, among multiple SPs, and corresponding settings under heterogeneous network services.

\section*{Acknowledgment}

The research was supported by National Basic Research Program of China (973 Program) through grant 2012CB316004, 100 Talents Program of Chinese Academy of Sciences and MIIT of China through grant 2011ZX03001-006-01.

\bibliographystyle{IEEEtran}
\bibliography{reference}

\end{document}